\documentclass[preprint]{elsarticle}   
\usepackage{amsmath,amssymb}          
\usepackage{amsfonts,amsthm}
\usepackage{graphicx}
\usepackage{float}
\usepackage{booktabs}
\usepackage{tikz}
\usepackage{xcolor}
\usepackage[expansion=false]{microtype}   
\IfFileExists{orcidlink.sty}{\usepackage{orcidlink}}{\providecommand{\orcidlink}[1]{}}
\usepackage{caption}   
\usepackage{hyperref}
\graphicspath{{figures/}}

\newcommand{\figscale}{1.0}  
\newcommand{\tabscale}{1.0}  
\newtheorem{theorem}{Theorem}

\begin{document}

\begin{frontmatter}

\title{Differential-Embedding Reconstruction of Dynamical Systems from Scalar Time Series}

\author[ntu]{Ameir Shaa\corref{cor1}\,\orcidlink{0000-0001-6839-3195}}
\ead{ameirshaa.akberali@ntu.edu.sg}
\cortext[cor1]{Corresponding author.}
\author[ntu]{Claude Guet}
\affiliation[ntu]{organization={School of Physical and Mathematical Sciences,
    Nanyang Technological University},
  addressline={21 Nanyang Link},
  postcode={637371},
  country={Singapore}}

\begin{abstract}

We study the reconstruction of an unknown dynamical system from a single noisy scalar time series. 
The goal is to recover the underlying dynamics for forecasting. 
We introduce a method that uses differential embedding coordinates to identify a rational closure of 
the embedding dynamics directly from data. 
The closure is identified through a weak-form regression pipeline, which avoids unstable pointwise
differentiation of noisy data. When applied to noise-free Lorenz and R\"ossler systems, the method recovers closures that support long forecasts across a broad ensemble of realizations ($18.1$ and $7.1$ Lyapunov times respectively). Under $15$--$30\%$ additive Gaussian noise, performance becomes system-dependent.
For the Lorenz system, forecast horizons remain short even in the best cases, whereas the R\"ossler system generally
performs better in absolute terms, though not once normalized by the Lyapunov time. Our proposed method recovers directly interpretable closure
coefficients which we compared against the known analytic closures of the Lorenz and R\"ossler systems.

\end{abstract}


\end{frontmatter}

\section{Introduction}

Many physical systems are observed through a single measured signal rather than through their full state. 
Recovering a predictive model from such data is a long-standing problem in nonlinear dynamics.

State-space reconstruction from a scalar observable has a long history. Classical
embedding theory~\cite{whitney1936differentiable} shows that smooth manifolds of
dimension~$m$ admit smooth embeddings in Euclidean space of dimension at most
$2m+1$. Packard
et al.~\cite{packard1980geometry} showed empirically that a single measured
coordinate can reconstruct the underlying dynamics, using both delay and
derivative coordinates and Takens' embedding theorem~\cite{takens1981detecting}
later provided the generic justification. Sauer et al.~\cite{sauer1991embedology} subsequently generalized beyond smooth
manifolds to attractors of finite box-counting dimension.
In practice, reconstruction and forecasting quality depends strongly on the delay, the embedding
dimension, the noise level, and the measurement function~\cite{fraser1986independent,kennel1992determining,casdagli1991state,farmer1987predicting,sugihara1990nonlinear}.
Differential embeddings belong to this same tradition. Letellier, Aguirre, and Maquet~\cite{letellier2005relation,aguirre2018structural} 
made explicit that reconstruction quality depends on the chosen observable and can
deteriorate sharply near low-observability regions.  Lainscsek and Gonzalez \cite{lainscsek2015delay,gonzalez2020assessing} provided a route through delay differential
analysis in which scalar time series are fit directly by low-order delay-differential models and
used to assess observability from data.

More recent work pursues learned coordinates and learned dynamics from partial
observations. Autoencoder-based methods reconstruct latent attractors from noisy
scalar series or jointly discover latent coordinates and sparse governing
equations~\cite{jiang2017state,gilpin2020deep,champion2019data}, and have
recently been extended to the single-scalar setting via deep delay
autoencoders~\cite{bakarji2023discovering}. Delay-coordinate
models such as HAVOK yield a linear-plus-forcing representation in delay
space~\cite{brunton2017havok}, while reservoir-computing and deep-learning
approaches can replicate attractor statistics and forecast partially observed
chaotic systems without explicit closed-form
equations~\cite{pathak2017replicate,gupta2022partially,wang2021reconstructing}.
In parallel, sparse-regression methods such as SINDy and its rational, implicit,
and weak-form variants identify compact and interpretable models directly from
data, including systems with rational nonlinearities and noisy
measurements~\cite{brunton2016discovering,mangan2016inferring,messenger2021weak,kaheman2020sindy}.
The present work is most similar to this last line of research but differs in what it recovers. We
identify an explicit rational closure in differential-embedding coordinates from a single scalar
observable, estimating its denominator and numerator in two stages and evaluating the result by
forecast skill. Because the recovered coefficients are interpretable, we compare them directly
with the analytic closures of the test systems, separating forecast success from coefficient accuracy.

We formulate the reconstruction problem in differential embedding coordinates, using successive time 
derivatives of the scalar signal as coordinates. This is the differential analog of the more familiar 
delay embedding~\cite{sauer1991embedology}. In these coordinates, we seek a closed-form description of the 
dynamics through a rational closure, in which the highest derivative is represented as a ratio of polynomial functions. 
Rational models of this kind have been identified from full-state measurements by rational 
SINDy~\cite{mangan2016inferring}. Here we extend that perspective to the scalar-observable setting, where the 
denominator is not known in advance and must be recovered along with the numerator.

The paper has three main contributions. First, we introduce a method for recovering a rational closure 
directly from noisy scalar data in differential embedding coordinates, without prior knowledge of the governing 
equations or access to full-state observations. Second, the recovered model is explicit and interpretable, 
which allows coefficient-level validation on test systems whose analytic closures are known, rather than 
relying on forecast scores alone. Third, we provide a comparative study on Lorenz and R\"ossler systems under clean and 
noisy observations, documenting how forecast quality and coefficient recovery vary across systems and noise levels.

Section~\ref{sec:theory} develops the closure in differential embedding coordinates. Section~\ref{sec:method} 
describes the reconstruction methodology. Section~\ref{sec:results} reports results on 
Lorenz~\cite{lorenz1963deterministic} and R\"ossler~\cite{rossler1976equation} systems under clean and noisy observations. 
Section~\ref{sec:discussion} discusses closure-direction diagnostics,
observability, comparison with prior methods, and limitations.

\section{Closure in Differential Embedding Coordinates}
\label{sec:theory}
Consider a uniformly sampled scalar time series $\{s_i\}_{i=1}^{N}$, where
$s_i = u(t_i)$, $t_i \in [0,T]$, and the sampling interval is
$\Delta t = t_{i+1}-t_i$.
Here, $u:[0,T]\to\mathbb{R}$ is assumed to be $k$-times differentiable.
Let $s_{\Delta t}(t)$ denote an interpolation of the sampled data.
In the continuum limit $\Delta t \to 0$, we assume that
$s_{\Delta t}(t) \to u(t)$.
The interpolation procedure is described in greater detail in
Section~\ref{sec:method}.

We further assume that $u(t)$ is a scalar observable of an unknown smooth
autonomous dynamical system,
\begin{equation}
\dot{\mathbf{U}} = \mathbf{F}(\mathbf{U}),
\qquad
\mathbf{U} \in \mathbb{R}^n,
\label{eq:system}
\end{equation}
where $\mathbf{F}:\mathbb{R}^n\to\mathbb{R}^n$ is smooth but otherwise unknown.
The system is assumed to possess a compact invariant attractor
$\mathcal{A}\subset\mathbb{R}^n$, and the trajectory of interest satisfies
$\mathbf{U}(t)\in\mathcal{A}$ for all $t$ under consideration.
The observed signal is given by
\begin{equation}
u(t) = \pi\bigl(\mathbf{U}(t)\bigr),
\qquad
\pi : \mathbb{R}^n \to \mathbb{R},
\label{eq:scalar_obs}
\end{equation}
where $\pi$ is a smooth observation function.
Hence, although the dynamical system is defined on the ambient state space
$\mathbb{R}^n$, the dynamics considered here are those restricted to
$\mathcal{A}$.

The successive time derivatives of $u(t)$ define differential embedding coordinates
$u_j(t)$, $j = 0,\ldots,k$, collected into the differential-embedding map
\begin{equation}
\pi_J : \mathbb{R}^n \to \mathbb{R}^{k+1}, \quad
\pi_J(\mathbf{U}) = \bigl(u,\,\dot{u},\,\ldots,\,u^{(k)}\bigr).
\label{eq:jet_map}
\end{equation}
Let $\mathcal{A}_J=\pi_J(\mathcal{A})\subset\mathbb{R}^{k+1}$ denote the embedded
attractor in differential embedding coordinates.
Wherever $\pi_J$ is locally invertible, $\mathcal{A}$ and $\mathcal{A}_J$ are
locally diffeomorphic, and differentiating $\mathbf{u}=\pi_J(\mathbf{U})$ along
trajectories gives
\begin{equation}
\dot{\mathbf{u}} = J_{\pi_J(\mathbf{U})}\,\mathbf{F}(\mathbf{U}).
\label{eq:jet_chain}
\end{equation}
where $J_{\pi_J(\mathbf{U})}$ denotes the Jacobian matrix of the differential-embedding map $\pi_J$ at $\mathbf{U}$.
As proved in \ref{app:jet_proof},  there exists a smooth closure function $\Phi$ such that
\begin{equation}
\dot{u}_i = u_{i+1}, \quad i = 0,\ldots,k-1, \qquad
\dot{u}_k = \Phi(u_0,\ldots,u_k).
\label{eq:jet_ode}
\end{equation}

Therefore, there exists a system of differential equations on
$\mathcal{A}_J$ that is smoothly conjugate to the dynamical system
defined by~\eqref{eq:system}:
\begin{equation}
\frac{d}{dt}
\begin{pmatrix}
u_0 \\ u_1 \\ \vdots \\ u_k
\end{pmatrix}
=
\begin{pmatrix}
u_1 \\ u_2 \\ \vdots \\ u_{k+1} = \Phi(u_0,\ldots,u_k)
\end{pmatrix}.
\label{eq:systemdiffeo}
\end{equation}

As we prove in \ref{app:sw}, the closure, $\Phi$, is rational by construction.  We can write $\Phi(\mathbf u)=N(\mathbf u)/D(\mathbf u)$ where $N(\mathbf u)$ and $D(\mathbf u)$ are represented as polynomials by virtue of the Stone-Weierstrass theorem~\cite{rudin1976principles} which states that for a compact and closed set, any continous function is arbitrarily well approximated by polynomials. This rational structure of closures in differential
embedding coordinates is well documented~\cite{letellier2005relation,lainscsek2015delay}.

Where the denominator $D(\mathbf u)$ vanishes (i.e. $D(\mathbf u) = 0$) is where the map $\pi_J$ loses invertibility and observability~\cite{letellier2005relation,aguirre2018structural}. In this paper, we term this locus the singular set (or pole) of the closure.

Our differential system then becomes

\begin{equation}
\frac{d}{dt}
\begin{pmatrix}
u_0 \\ u_1 \\ \vdots \\ u_k
\end{pmatrix}
=
\begin{pmatrix}
u_1 \\ u_2 \\ \vdots \\ u_{k+1}=\dfrac{N(u_0,\ldots,u_k)}{D(u_0,\ldots,u_k)}
\end{pmatrix}.
\label{eq:systemdiffeorational}
\end{equation}

The next section discusses in detail the methodology for determining the dimension of the system (ie. k), the approximation to making $s_i$ a smooth continous function of time and the polynomial approximation to N and D.

\section{Methodology}
\label{sec:method}

We seek to reconstruct a forecastable rational closure from a single noisy
scalar record. The methodology proceeds through a fixed sequence of stages:
spline-based differential embedding, weak-form regression, closure
identification, and forecast-based validation. We start from a
uniformly sampled scalar record $\{y_i\}_{i=1}^{N_{\mathrm{s}}}$ of $N_{\mathrm{s}}$ samples, where
$y_i = s(t_i) + \eta_i$ is a noisy observation of an underlying signal $s(t)$ at
sample time $t_i$ for $i=1,\ldots,N_{\mathrm{s}}$, with uniform spacing
$\Delta t = t_{i+1}-t_i$ and independent Gaussian observation noise
$\eta_i \stackrel{\mathrm{i.i.d.}}{\sim} \mathcal{N}(0,\sigma_\eta^2)$.

Differential embedding coordinates are then
estimated with a Reinsch spline \cite{reinsch1967smoothing}. From these coordinates,
weak-form regression matrices are assembled, the rational closure is identified, and the recovered
closure is forecast forward in time. Valid prediction time (VPT) is then used
to score forecast quality. Figure~\ref{fig:pipeline} summarizes this workflow.

The embedding dimension is assessed independently with the false nearest
neighbours (FNN) criterion~\cite{kennel1992determining}, using a delay $\tau$
chosen at the first local minimum of the average mutual
information~\cite{fraser1986independent}; this provides a data-driven estimate of
a sufficient embedding dimension, so the true system dimension need not be assumed
as prior information. We note that FNN and delay-coordinate diagnostics serve only as a check that three
coordinates suffice. For clarity and without loss of generality, we present the methodology in
a three-dimensional setting (i.e., $\mathbf{u} = (u_0,u_1,u_2) \in \mathbb{R}^3$);
a higher embedding dimension only enlarges the coordinate vector and the
regression matrices, leaving the construction unchanged.

A quintic Reinsch spline~\cite{reinsch1967smoothing} fits $u_0$
to the scalar record by minimising
\begin{equation}
  \sum_{i=1}^{N_{\mathrm{s}}} \bigl(y_i - u_0(t_i)\bigr)^2
  + \lambda \int \bigl(u_0^{(3)}\bigr)^2\,\mathrm{d}t,
\label{eq:reinsch}
\end{equation}
where the sum runs over the $N_{\mathrm{s}}$ samples, the integral is taken
over the sampling window $[t_1,\,t_{N_{\mathrm{s}}}]$, and the smoothing parameter $\lambda$
controls the interpolation--smoothing tradeoff.
The derivative coordinates $u_1$ and $u_2$ are then obtained by
analytically differentiating the fitted quintic spline.
Motivated by Reinsch's residual-target criterion~\cite{reinsch1967smoothing}, 
we set the penalty weight $\lambda = N_{\mathrm{s}}\hat\sigma^2$,
where the noise scale $\hat\sigma$ is estimated from the data as described below.

In the clean data limit, $\hat\sigma \to 0$ so $\lambda \to 0$ and the spline
interpolates the observed signal. In the noisy regime, $\hat\sigma$
is estimated from the second differences
$\Delta^2 y_i = y_{i+1} - 2y_i + y_{i-1}$, which suppress the smooth signal,
leaving behind a noise-dominated sequence. Following Donoho and
Johnstone~\cite{donoho1994ideal}, we use the median absolute deviation (MAD) of
these second differences,
$\mathrm{MAD}(z) = \mathrm{median}_i\,|z_i - \mathrm{median}_j\,z_j|$, as an
outlier-resistant scale,
\begin{equation}
  \hat\sigma = \frac{1.4826\,\mathrm{MAD}(\Delta^2 y)}{\sqrt{6}}.
\label{eq:noise_scale}
\end{equation}
The factor $1.4826$ converts the MAD into a standard deviation for Gaussian data,
and the $\sqrt{6}$ corrects for the noise being amplified by the second-difference
operator; together they turn the robust spread of $\Delta^2 y$ into an estimate of
the noise level $\sigma_\eta$ (\ref{app:noise_scale}). In this way the true
noise level need not be supplied in advance.

Given differential embedding coordinates $(u_0,u_1,u_2)$, we seek a rational
closure satisfying Eq.~\eqref{eq:closure}. The denominator and numerator are
represented in monomial libraries $\Phi_D$ and $\Phi_N$ of maximum total degree
$d_D$ and $d_N$,
\begin{equation}
\begin{split}
D(\mathbf{u}) &= \theta_D^\top \Phi_D(\mathbf{u})
= \!\!\sum_{|\beta|\le d_D}\!\! \theta_{D,\beta}\,
u_0^{\beta_0}u_1^{\beta_1}u_2^{\beta_2},\\
N(\mathbf{u}) &= \theta_N^\top \Phi_N(\mathbf{u})
= \!\!\sum_{|\alpha|\le d_N}\!\! \theta_{N,\alpha}\,
u_0^{\alpha_0}u_1^{\alpha_1}u_2^{\alpha_2},
\end{split}
\label{eq:libraries}
\end{equation}
where $\Phi_D$ and $\Phi_N$ span all monomials in $(u_0,u_1,u_2)$ up to total
degree $d_D$ and $d_N$ respectively and $\theta_{\gamma} \quad \gamma \in \{N,D\}$ are the coefficients. 
Note that both the denominator degree
$d_D$ and the numerator degree $d_N$ are free hyperparameters of the closure family.
The coefficients of $D(\mathbf{u})$ within a chosen degree-$d_D$ library are
recovered from the data and determine the singular set.

The closure can be identified either in strong form, using pointwise samples,
or in weak form, using windowed integral constraints. In strong form, the
closure identity
\begin{equation}
D(\mathbf{u})\,\dot{u}_2 = N(\mathbf{u})
\label{eq:strong_identity}
\end{equation}
is enforced pointwise along the trajectory.
Writing $u_3 \equiv \dot{u}_2 = s^{(3)}$ for the highest derivative
(obtained by differentiating the fitted spline once more), the resulting
design matrix is
\begin{equation}
  M = \bigl[\,u_3 \odot \Phi_D \;\; -\Phi_N\,\bigr],
\label{eq:strong_form}
\end{equation}
where $\Phi_D$ and $\Phi_N$ are the libraries for $D$ and $N$ evaluated
row-wise over the samples, and $\odot$ denotes the elementwise product of the
sample vector $u_3$ with each column of $\Phi_D$.

Stacking Eq.~\eqref{eq:strong_identity} over the samples gives a homogeneous
system $M\,\theta = 0$, with $\theta=[\theta_D^\top,\theta_N^\top]^\top$.
As derived in \ref{app:nullspace}, the
coefficients we seek are therefore a null vector of $M$: a solution is fixed only
up to an overall scale, so the closure is a single \emph{direction} in coefficient
space rather than an isolated point. With noisy data $M$ has no exact null vector,
so we take the direction it maps closest to zero---the right singular vector of its
smallest singular value~\cite{mangan2016inferring,kaheman2020sindy}, which
minimizes $\|M\theta\|_2$ over unit-norm $\theta$. Because the entire model is
carried by this one direction, its accuracy is limited by how well that direction
is resolved. The resolution is set by the gap between the smallest singular value
$\sigma_{\min}$ and the next, $\sigma_{\mathrm{next}}$: the estimate's sensitivity
to perturbations of $M$ scales as $1/(\sigma_{\mathrm{next}}-\sigma_{\min})$, so
when the two are close a small perturbation of $M$ rotates the recovered direction
far from the true one.

Noise enters $M$ most strongly through the $u_3\odot\Phi_D$ block. The lower-order
coordinates $u_0,u_1,u_2$ that populate both libraries are also spline estimates,
but they are comparatively smooth, whereas $u_3=s^{(3)}$ is a third derivative and
differentiation amplifies high-frequency content~\cite{chartrand2011numerical}.
On clean data the spline interpolates the signal, $u_3$ is accurate, and the
null vector is recovered reliably. Under noise the high-frequency error in $u_3$
perturbs this block, narrows the singular-value gap, and rotates the estimated
null vector away from the true coefficients. This motivates the weak
formulation.

The weak-form reformulation~\cite{schaeffer2017sparse,messenger2021weak}
replaces pointwise use of $u_3$ by integration against a smooth, compactly supported
test function $\psi$ over a window $[a,b]$,
\begin{equation}
\int_a^b D(\mathbf{u}(t))\,\dot{u}_2(t)\,\psi(t)\,\mathrm{d}t
=
\int_a^b N(\mathbf{u}(t))\,\psi(t)\,\mathrm{d}t.
\label{eq:weak_form}
\end{equation}
Integrating the left-hand side by parts moves the time derivative off $u_2$.
Since $\dot{D} = (\partial_{u_0}D)\,u_1 + (\partial_{u_1}D)\,u_2
+ (\partial_{u_2}D)\,u_3$, the highest derivative $u_3$ is eliminated
precisely when $D$ is independent of $u_2$. 
\begin{equation}
\begin{split}
\!-\!\!\int_a^b u_2 \Bigl(
&D\,\psi'
+ (\partial_{u_0}D)\,u_1\,\psi \\
&\qquad
+ (\partial_{u_1}D)\,u_2\,\psi
\Bigr)\,\mathrm{d}t
=
\int_a^b N(\mathbf{u})\,\psi\,\mathrm{d}t,
\end{split}
\label{eq:ibp_general}
\end{equation}
where $\psi$ is explicitly chosen such that $\psi(a)=\psi(b)=0$ and 
the boundary terms vanish by construction. The weak form therefore restricts
$\Phi_D$ to monomials in $(u_0,u_1)$, while $\Phi_N$ remains the full
degree-$d_N$ library in $(u_0,u_1,u_2)$. This restriction is a modeling choice,
adopted so that the regression avoids pointwise use of $u_3$. In general, it is not the case
that every true closure must have $D$ independent of $u_2$.

The weak-form window parameters are free hyperparameters of the closure family along with the aforementioned 
denominator degree $d_D$ and the numerator degree $d_N$.

Evaluating the weak form~\eqref{eq:ibp_general} against a collection of test
functions $\{\psi_k\}_{k=1}^K$, each supported on its own window, and stacking
the resulting $K$ equations yields the homogeneous linear system
\begin{equation}
\bigl[\,M_D \;\; -M_N\,\bigr]
\begin{bmatrix}
\theta_D \\ \theta_N
\end{bmatrix}
= 0,
\label{eq:weak_system}
\end{equation}
where the matrix entries are
\begin{align}
(M_D)_{kj}
&=
-\int u_2 \Bigl(
\phi_{D,j}\,\psi_k'
+ (\partial_{u_0}\phi_{D,j})\,u_1\,\psi_k \notag\\
&\qquad\qquad
+ (\partial_{u_1}\phi_{D,j})\,u_2\,\psi_k
\Bigr)\,\mathrm{d}t,
\label{eq:MD_entries} \\
(M_N)_{k\ell}
&=
\int \phi_{N,\ell}(\mathbf{u})\,\psi_k\,\mathrm{d}t,
\label{eq:MN_entries}
\end{align}
with $\phi_{D,j}\in\Phi_D$ and $\phi_{N,\ell}\in\Phi_N$.

The weak-form matrices $(M_D,M_N)$ are used in a two-stage recovery, with the
denominator estimated first because the singular-set geometry it encodes is the
fragile part of the closure.
Both stages operate on the column-normalized joint
matrix $[\,M_D\;\;-M_N\,]$ of Eq.~\eqref{eq:weak_system}, whose smallest right
singular vector $[\theta_D^\top,\theta_N^\top]^\top$ is the null-space estimate;
only the denominator block $\theta_D$ is retained at this stage. In the clean
limit a single global SVD of $[\,M_D\;\;-M_N\,]$ provides $\theta_D$ directly.
Under noise the trajectory is divided into $12$ overlapping local windows
($120$ samples each), and a separate joint SVD on each window yields a local
denominator estimate. Because each is defined only up to sign and scale, the
local estimates are unit-normalized, sign-aligned so that the coefficient of
$u_0$ is positive, averaged componentwise, and renormalized to give the pooled
denominator~\cite{fasel2022ensemble}. These $12$ averaging windows are distinct
from the $K$ test-function windows used to assemble $(M_D,M_N)$.
Once the denominator has been fixed, the numerator coefficients are obtained from
the linear system
\begin{equation}
M_N\,\theta_N = M_D\,\theta_D,
\label{eq:pin_system}
\end{equation}
solved by ordinary least squares over the test-function windows.

The identified closure defines the autonomous system
$\dot{u}_0=u_1$, $\dot{u}_1=u_2$, and
$\dot{u}_2 = N(\mathbf{u})/D(\mathbf{u})$.
Forecasts are integrated numerically with DOP853~\cite{hairer1993solving}. The recovered closure is
integrated directly, without additional refinement or postprocessing.

Forecast accuracy is measured by the valid prediction time (VPT), the elapsed
time over which a forecast tracks the truth. Forecasts are scored on $t\in[0,H]$
(duration $H$; $H/\Delta t+1$ samples at spacing $\Delta t$). If no failure
occurs within this window, VPT is reported as $H$. Over the scoring window,
the normalized pointwise error is
\begin{equation}
  e(t)=
  \frac{|u_0^{\mathrm{pred}}(t) - u_0^{\mathrm{true}}(t)|}
       {\mathrm{std}(u_0^{\mathrm{true}}\vert_{[0,H]})},
\label{eq:vpt}
\end{equation}
and the prediction is regarded as valid up to the onset of the first run of
$20$ consecutive samples with $e(t) > \varepsilon_{\mathrm{thr}}$, where the error
threshold is $\varepsilon_{\mathrm{thr}}=0.2$. In other words, isolated threshold crossings are
ignored unless they persist for $20$ consecutive samples.

For each realization, we test the recovered closure by forecasting from eight
starting points spaced evenly across the interior of the trajectory, scoring
each forecast by its valid prediction time. All eight points lie on the same
trajectory used for identification, so this is a within-trajectory assessment of
the recovered dynamics: we restart the model at eight points along the
identification trajectory and forecast forward from each. These eight windows
probe robustness within a single trajectory; they are distinct from independent noise realizations.
Each realization is summarized by the best valid prediction time over its eight
initial-condition windows.

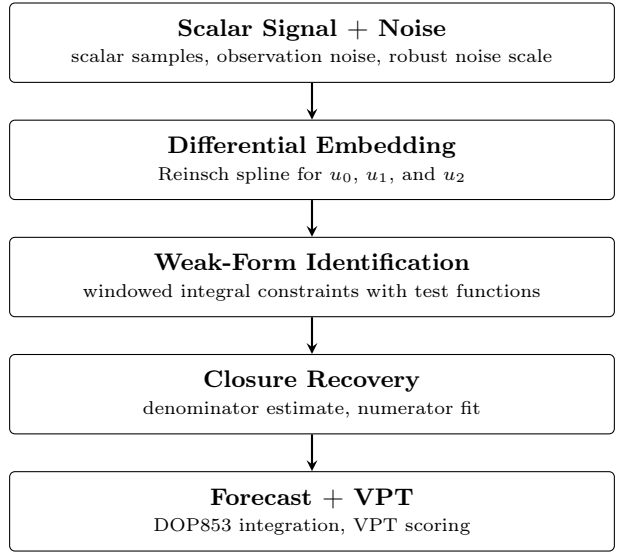
\begin{figure}[H]
\centering
\scalebox{\figscale}{%
\begin{tikzpicture}[
  stage/.style={draw, rounded corners=2pt, minimum width=8cm,
    minimum height=1.05cm, align=center, font=\small},
  arr/.style={-stealth, thick},
  every node/.style={inner sep=4pt}
]
\node[stage] (s1) at (0,0)
  {\textbf{Scalar Signal + Noise}\\[-1pt]
   \scriptsize scalar samples, observation noise, robust noise scale};

\node[stage] (s2) at (0,-1.55)
  {\textbf{Differential Embedding}\\[-1pt]
   \scriptsize Reinsch spline for $u_0$, $u_1$, and $u_2$};

\node[stage] (s3) at (0,-3.1)
  {\textbf{Weak-Form Identification}\\[-1pt]
   \scriptsize windowed integral constraints with test functions};

\node[stage] (s4) at (0,-4.65)
  {\textbf{Closure Recovery}\\[-1pt]
   \scriptsize denominator estimate, numerator fit};

\node[stage] (s5) at (0,-6.2)
  {\textbf{Forecast + VPT}\\[-1pt]
   \scriptsize DOP853 integration, VPT scoring};

\draw[arr] (s1) -- (s2);
\draw[arr] (s2) -- (s3);
\draw[arr] (s3) -- (s4);
\draw[arr] (s4) -- (s5);
\end{tikzpicture}%
}
\caption{Reconstruction pipeline for the reported method. A noisy scalar signal
  is converted into differential embedding coordinates, used to assemble
  weak-form regression matrices, and then mapped to a rational closure whose
  forecast skill is evaluated by VPT.}
\label{fig:pipeline}
\end{figure}

\section{Results}
\label{sec:results}

We report results for the Lorenz~\cite{lorenz1963deterministic} and
R\"ossler~\cite{rossler1976equation} systems, each observed through the single
scalar signal $s=x$; the full set of hyperparameters is collected in
Table~\ref{tab:settings} (\ref{app:settings}). Across realizations the
model family, spline procedure, weak-form matrices, and scoring rule are held
fixed; only the noise realization, and hence the recovered closure, varies. In each
realization one closure is identified from the full scalar record and evaluated on
eight interior initial-condition windows drawn from the same trajectory, so
identification and forecasting use the same simulated trajectory.

Figure~\ref{fig:fnn} shows the FNN fractions for both systems, with delays chosen
from the first local minimum of the average mutual
information~\cite{fraser1986independent}. In both cases the FNN fraction drops
below $1\%$ at $m=3$, consistent with the three-coordinate reconstruction used here.

\begin{figure}[H]
\centering
\scalebox{1.0}{%
\includegraphics[width=0.95\textwidth]{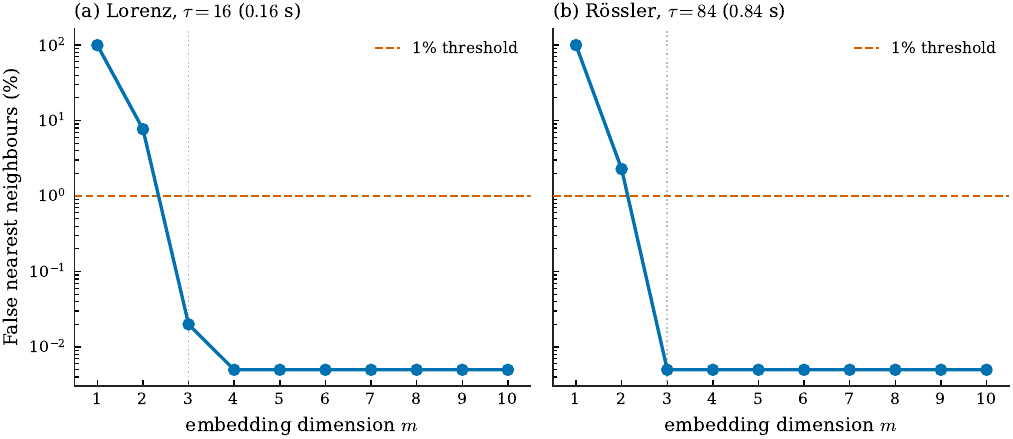}%
}
\caption{False nearest neighbours (FNN) fraction versus embedding dimension
  $m$ for Lorenz ($\tau=16$, i.e.\ $0.16$ in Lorenz time) and R\"{o}ssler ($\tau=84$,
  i.e.\ $0.84$ in R\"{o}ssler time), with delays chosen at the first local minimum of the
  average mutual information. In both cases, the FNN fraction falls below
  $1\%$ (dashed red line) at $m=3$.}
\label{fig:fnn}
\end{figure}

As mentioned in Section~\ref{sec:method}, the denominator and numerator degrees ($d_D$ and $d_N$ respectively) 
along with the weak-form parameters are free hyperparameters of the closure family. 
We tune these hyperparameters by maximizing the invariant metric VPT as defined in Eq.~\eqref{eq:vpt}.
The VPT metric compares a forecast with the measured continuation of the signal which requires no knowledge 
of the governing equations. The VPT metric is invariant because the
error is normalized by the signal's own standard deviation. 

Both benchmark systems are dimensionless, so their model time carries no physical
unit; we refer to it as \emph{Lorenz time} and \emph{R\"ossler time} respectively.
The two systems evolve on very different intrinsic time scales, so model time is
not comparable between them. In this section we therefore report forecast horizons
both in model time and in Lyapunov times $\Lambda=\lambda_1 T$, where $\lambda_1$
is the largest Lyapunov exponent of the system. Perturbations grow approximately as
$e^{\lambda_1 t}$~\cite{wolf1985determining}, so R\"ossler's much smaller
$\lambda_1$ means that a forecast that is longer in R\"ossler time can still span
fewer e-folding times of error growth: in model time R\"ossler forecasts longer
than Lorenz at every noise level, whereas in Lyapunov times the ordering reverses.
Unless otherwise stated, elsewhere in the paper we quote Lyapunov times only, as they are the measure that compares across systems.

\subsection{Noise-free data}
\label{sec:results_clean}

In the noise-free case, the selected degrees for the Lorenz and Rossler systems were
($d_D=1$, $d_N=4$) and ($d_D=1$, $d_N=3$) respectively. These coincide with the degrees of the analytic Lorenz and R\"ossler
systems. We note that this is an outcome of the selection rather than an input to it: the
degrees were chosen by maximizing forecast skill, without reference to the analytic
closure.

Let $\hat N$ and $\hat D$ denote the recovered numerator and denominator
coefficient vectors and $N^\star,D^\star$ their analytic counterparts, both
expressed as coefficients over the same set of monomials so that corresponding
terms can be compared directly. 

For each system, we report the results with respect to the realization achieving the highest VPT.
A rational closure is unchanged under a common rescaling
$(N,D)\mapsto(\kappa N,\kappa D)$, so $\hat N$ and $\hat D$ are determined only up
to a single scale factor $\kappa$ (\ref{app:nullspace}). The same gauge
freedom arises in implicit and rational sparse identification, where the model
is likewise recovered from a null space~\cite{mangan2016inferring,kaheman2020sindy}.
Comparison with the analytic closure therefore requires a choice of gauge; we fix
$\kappa$ by a least-squares fit of $\kappa\hat D$ to $D^\star$, so that the
recovered coefficients are expressed on the analytic scale.

For both the numerator and the denominator we report the relative $\mathrm{L}^2$
error
\begin{equation}
\frac{\lVert \kappa\hat X - X^\star \rVert_2}{\lVert X^\star \rVert_2},
\qquad X\in\{N,D\},
\label{eq:rel_l2}
\end{equation}
where $\lVert\cdot\rVert_2$ is the Euclidean norm on the coefficient vectors over
the common monomial basis.

In the absence of observation noise the closure is recovered essentially exactly.
Tables~\ref{tab:coeff_clean_lorenz} and~\ref{tab:coeff_clean_rossler} compare
$\kappa\hat N$ with $N^\star$ for each system:
the relative errors are of order $10^{-5}$ on every monomial, and the denominator
relative $\mathrm{L}^2$ errors are $5.0\times10^{-5}$ (Lorenz) and
$2.4\times10^{-7}$ (R\"ossler). The recovered scale reproduces the analytic
denominator exactly ($\kappa=3$ for Lorenz, where $D^\star=3u_0$). Exact recovery at
the selected degree confirms that the method identifies the closure itself, not
merely a forecast-capable surrogate.

\begin{table}[H]
\centering\small
\caption{Clean-data ($\sigma_{\mathrm{frac}}=0$) coefficient recovery for the
  best-forecast Lorenz realization. $N^\star$ is the analytic numerator
  coefficient; Rel.\ $\mathrm{L}^2$ is each monomial's contribution to the
  numerator relative $\mathrm{L}^2$ error, $|\kappa\hat N_i-N^\star_i|/\lVert
  N^\star\rVert_2$, so rows sum in quadrature to the total below. $\kappa=3$ is
  fixed by a least-squares fit of $\kappa\hat D$ to $D^\star=3u_0$.}
\label{tab:coeff_clean_lorenz}
\scalebox{\tabscale}{%
\begin{tabular}{lrr}
\toprule
Monomial & $N^\star$ & Rel.\ $\mathrm{L}^2$ \\
\midrule
$u_1 u_2$    & $3$    & $4.2\times10^{-8}$ \\
$u_1^{2}$    & $33$   & $1.1\times10^{-6}$ \\
$u_0 u_2$    & $-41$  & $7.2\times10^{-8}$ \\
$u_0 u_1$    & $-88$  & $9.8\times10^{-6}$ \\
$u_0^{2}$    & $2160$ & $3.2\times10^{-5}$ \\
$u_0^{3}u_1$ & $-3$   & $1.5\times10^{-7}$ \\
$u_0^{4}$    & $-30$  & $4.0\times10^{-7}$ \\
\midrule
\multicolumn{3}{l}{\footnotesize Total $N$ rel.\ $\mathrm{L}^2\approx3.3\times10^{-5}$ \quad $D$ rel.\ $\mathrm{L}^2=5.0\times10^{-5}$} \\
\bottomrule
\end{tabular}%
}
\end{table}

\begin{table}[H]
\centering\small
\caption{Clean-data ($\sigma_{\mathrm{frac}}=0$) coefficient recovery for the
  best-forecast R\"ossler realization. Columns as in Table~\ref{tab:coeff_clean_lorenz};
  $\kappa=500$, fixed the same way as for Lorenz.}
\label{tab:coeff_clean_rossler}
\scalebox{\tabscale}{%
\begin{tabular}{lrr}
\toprule
Monomial & $N^\star$ & Rel.\ $\mathrm{L}^2$ \\
\midrule
$1$          & $-118$              & $3.4\times10^{-8}$ \\
$u_2$        & $1.622\times10^{4}$ & $5.2\times10^{-6}$ \\
$u_1$        & $-313$              & $9.7\times10^{-8}$ \\
$u_1 u_2$    & $500$               & $2.6\times10^{-7}$ \\
$u_1^{2}$    & $-100$              & $2.0\times10^{-8}$ \\
$u_0$        & $1.684\times10^{4}$ & $5.4\times10^{-6}$ \\
$u_0 u_2$    & $-5700$             & $2.8\times10^{-6}$ \\
$u_0 u_1$    & $1160$              & $4.7\times10^{-7}$ \\
$u_0^{2}$    & $-5800$             & $2.8\times10^{-6}$ \\
$u_0^{2}u_2$ & $500$               & $3.6\times10^{-7}$ \\
$u_0^{2}u_1$ & $-100$              & $1.9\times10^{-7}$ \\
$u_0^{3}$    & $500$               & $3.4\times10^{-7}$ \\
\midrule
\multicolumn{3}{l}{\footnotesize Total $N$ rel.\ $\mathrm{L}^2\approx8.4\times10^{-6}$ \quad $D$ rel.\ $\mathrm{L}^2=2.4\times10^{-7}$} \\
\bottomrule
\end{tabular}%
}
\end{table}

The clean forecasts are correspondingly long. For Lorenz, the best realization and
the top-five mean both reach the $20$ Lorenz-time scoring horizon ($18.1$ Lyapunov
times), while the median realization reaches $11.6$ Lorenz time ($10.5$ Lyapunov
times); this median is consistent with the Lyapunov-amplification ceiling set by
the integrator tolerance (\ref{app:tol}). For R\"ossler, the best
realization and the top-five mean both reach the $100$ R\"ossler-time horizon
($7.1$ Lyapunov times). Figure~\ref{fig:attractor}
shows that the recovered closures also reproduce the geometry of the clean
attractors under long rollout.

\begin{figure}[H]
\centering
\scalebox{1.0}{%
\includegraphics[width=\textwidth]{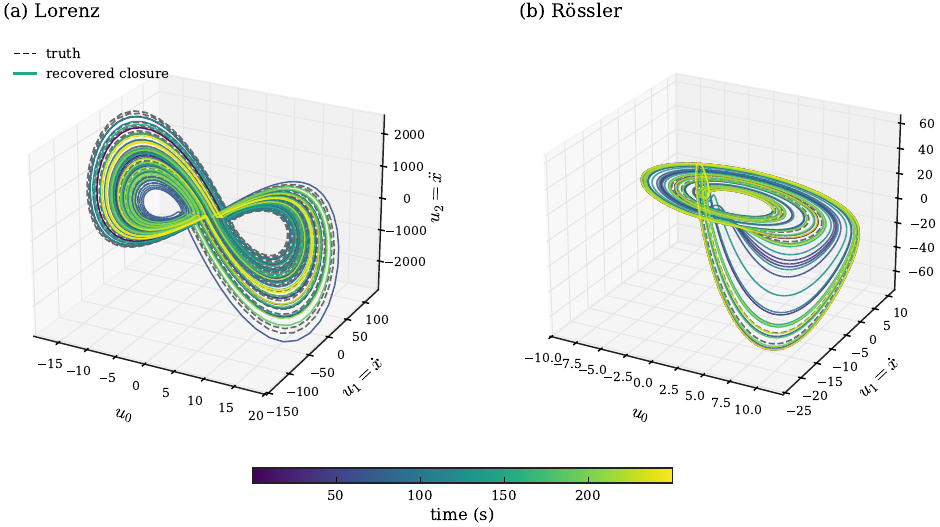}%
}
\caption{Three-dimensional differential-embedding $(u_0,u_1,u_2)$ attractors on
  clean data ($\sigma=0$): truth trajectory versus a long rollout of the
  recovered closure, shown for one successful realization on each
  system. The recovered closures reproduce the geometry of the clean
  attractors.}
\label{fig:attractor}
\end{figure}

\subsection{Noisy data}
\label{sec:results_noisy}

Under the addition of Gaussian noise, the degrees for the Lorenz and Rossler systems 
needed to be selected again by means of maximizing VPT as we did in the noise-free case. 
The selected degrees for the Lorenz and Rossler systems were 
($d_D=1$, $d_N=4$) and ($d_D=1$, $d_N=3$) respectively 
(as we had in the noise-free case).

However, under $15$--$30\%$ observation noise the two systems separate. For Lorenz
the best realization reaches $4.21$ Lorenz time ($3.81$ Lyapunov times) at
$\sigma_{\mathrm{frac}}=0.15$ and $4.33$ ($3.92$) at $0.30$, with top-five means of
$3.66$ ($3.32$) and $3.67$ ($3.32$). For R\"ossler the best realization reaches
$10.0$ R\"ossler time ($0.71$ Lyapunov times) and $8.61$ ($0.61$), with top-five
means of $7.78$ ($0.55$) and $7.83$ ($0.56$). Lorenz therefore loses most of its
clean-data horizon, while R\"ossler retains a larger fraction of its own.

Coefficient recovery does not survive the noise. The recovered null vector is
rotated away from the analytic direction rather than perturbed about it, so the
coefficients are not a noisy estimate of the true ones but a different direction in
coefficient space: relative errors reach order unity or larger on most monomials.
For the best-forecast realizations, the denominator relative $\mathrm{L}^2$ error
rises from $5.0\times10^{-5}$ (Lorenz) and $2.4\times10^{-7}$ (R\"ossler) on clean
data to $0.28$ and $0.21$ (Lorenz) and $0.94$ and $0.99$ (R\"ossler) at
$\sigma_{\mathrm{frac}}=0.15$ and $0.30$ respectively. A term-by-term comparison is
therefore uninformative in this regime. Coefficient accuracy also does not
correlate with forecast skill once noise is present
(Section~\ref{sec:cosine}). We therefore report forecasts here.

\begin{table}[H]
\centering
\small
\caption{Valid prediction time (VPT) for the recovered weak-form rational closure across the Lorenz and
R\"ossler benchmark systems per noise level. For each system and noise level we report, over the
$5{,}000$-realization sweep, the best VPT and the mean of the top five
realizations, each given in model time and in Lyapunov times $\Lambda=\lambda_1 T$
(largest Lyapunov exponents $\lambda_1=0.906$ for Lorenz and $0.071$ for
R\"ossler~\cite{sprott2003chaos}, in inverse model time). Each
realization is evaluated on eight initial-condition windows from the same
trajectory using the sustained threshold rule ($\varepsilon_{\mathrm{thr}}=0.2$, $20$
consecutive samples).}
\label{tab:vpt}
\resizebox{\linewidth}{!}{%
\begin{tabular}{llcccc}
\toprule
System / observable & $\sigma_{\mathrm{frac}}$ & Best (model time) & Top-5 (model time)
  & Best ($\Lambda$) & Top-5 ($\Lambda$) \\
\midrule
Lorenz $s=x$      & $0.00$ & $20.00$ & $20.00$ & $18.1$ & $18.1$ \\
Lorenz $s=x$      & $0.15$ & $4.21$  & $3.66$  & $3.81$ & $3.32$ \\
Lorenz $s=x$      & $0.30$ & $4.33$  & $3.67$  & $3.92$ & $3.32$ \\
\midrule
R\"{o}ssler $s=x$ & $0.00$ & $100.00$ & $100.00$ & $7.10$ & $7.10$ \\
R\"{o}ssler $s=x$ & $0.15$ & $10.00$ & $7.78$  & $0.71$ & $0.55$ \\
R\"{o}ssler $s=x$ & $0.30$ & $8.61$  & $7.83$  & $0.61$ & $0.56$ \\
\bottomrule
\end{tabular}%
}
\end{table}

Table~\ref{tab:vpt} collects the valid prediction times for both systems across all
noise levels, and Figures~\ref{fig:overlay} and~\ref{fig:overlay_rossler} overlay
the recovered forecast with the true $x(t)$ at each noise level for Lorenz and
R\"ossler respectively.

\begin{figure}[H]
\centering
\scalebox{1.0}{%
\includegraphics[width=0.88\textwidth]{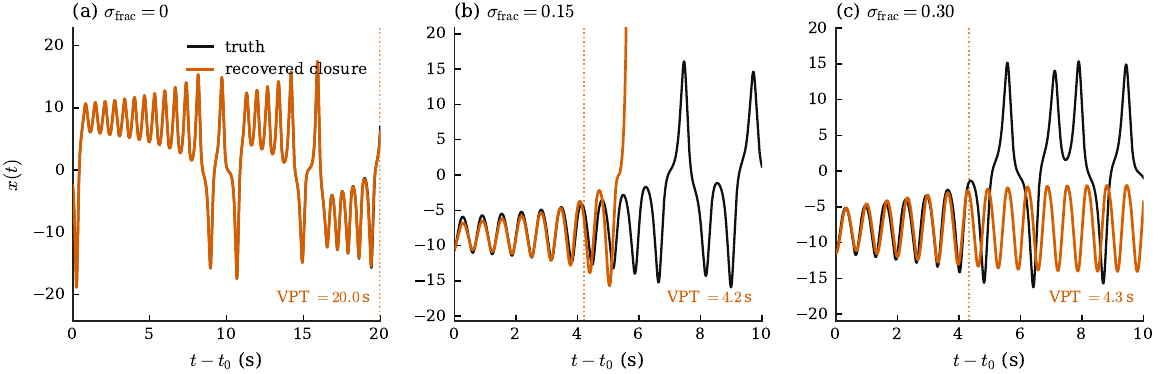}%
}
\caption{Lorenz prediction horizon overlaid with the true trajectory at each noise level.
  Each panel compares the recovered closure forecast with the true
  $x(t)$ and reports the corresponding valid prediction time (VPT). The clean
  case reaches the $20$ Lorenz-time horizon, while the noisy cases remain near
  $4$ Lorenz time ($\approx3.6$ Lyapunov times), consistent with Table~\ref{tab:vpt}.}
\label{fig:overlay}
\end{figure}

\begin{figure}[H]
\centering
\scalebox{1.0}{%
\includegraphics[width=0.88\textwidth]{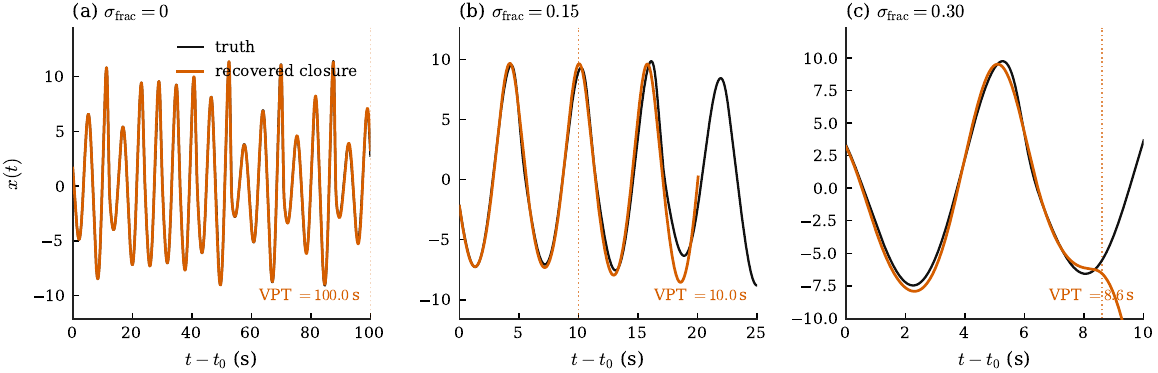}%
}
\caption{R\"ossler prediction horizon overlaid with the true trajectory at each noise level. Each panel compares the recovered closure forecast with the
  true $x(t)$ and reports the corresponding valid prediction time (VPT). The
  clean case reaches the $100$ R\"ossler-time horizon, while the noisy cases remain near
  $10$ R\"ossler time ($\approx0.7$ Lyapunov times), consistent with Table~\ref{tab:vpt}.}
\label{fig:overlay_rossler}
\end{figure}

\section{Discussion}
\label{sec:discussion}

\subsection{Closure-direction diagnostics}
\label{sec:cosine}

As reported in Section~\ref{sec:results_noisy}, under noise the recovered
coefficients are rotated away from the analytic direction, with relative errors
of order unity on most monomials; a term-by-term magnitude comparison is
therefore uninformative. A complementary, scale-invariant diagnostic is how well
the recovered closure reproduces the \emph{direction} of the analytic coefficient
vector. For two coefficient vectors $\mathbf{a}$ and
$\mathbf{b}$ we use the cosine similarity
\begin{equation}
\cos(\mathbf{a},\mathbf{b})
  = \frac{\mathbf{a}\cdot\mathbf{b}}
         {\lVert\mathbf{a}\rVert\,\lVert\mathbf{b}\rVert},
\label{eq:cosine}
\end{equation}
which lies in $[-1,1]$ and is invariant to a positive rescaling of either vector.
We write $\cos_N=\cos(\hat N,N^\star)$ for the numerator and
$\cos_D=\cos(\hat D,D^\star)$ for the denominator, where $\hat N,\hat D$ are the
recovered coefficients and $N^\star,D^\star$ the analytic ones in the common
monomial basis; both isolate whether the recovered closure points the right way in
coefficient space, independently of the scale factor $\kappa$.

Across the noisy Lorenz ensemble the denominator direction is reliably recovered
(median $\cos_D\approx0.93$), whereas the numerator direction is recovered cleanly
($\cos_N>0.9$) in only about a quarter of realizations and collapses
($\cos_N\approx0$) in most of the rest. Neither direction predicts forecast skill:
both are essentially uncorrelated with VPT (rank correlation $|\rho|\le0.04$ at both
noise levels), and the longest forecasts are not the best-recovered ones. At
$\sigma_{\mathrm{frac}}=0.15$, for instance, the single longest forecast has
$\cos_N\approx0.04$, and most of the twenty longest forecasts at each noise level
have $\cos_N<0.1$. Recovering the coefficient direction is thus neither necessary nor sufficient
for a long forecast. We restrict this analysis to Lorenz, for which a
mislocated singular set collapses the forecast under noise.

Direction accuracy therefore de-correlates from forecast skill, just as the
coefficient magnitudes do. Figure~\ref{fig:pole_decoupling} shows the underlying
geometry: the denominator direction is recovered ($\cos_D\approx0.93$) at every
forecast horizon, but for the affine denominator
$\bar D=\bar c_{00}+\bar c_{10}\,u_0$ the implied pole
$\hat p=-\bar c_{00}/\bar c_{10}$ scatters about its true value $0$
(standard deviation $\approx0.27$, roughly half sign reversals) even when the
direction is well recovered. In short, we find that pole location (as opposed to
closure direction) governs forecast skill.

\begin{figure}[H]
\centering
\scalebox{\figscale}{%
\includegraphics[width=\textwidth]{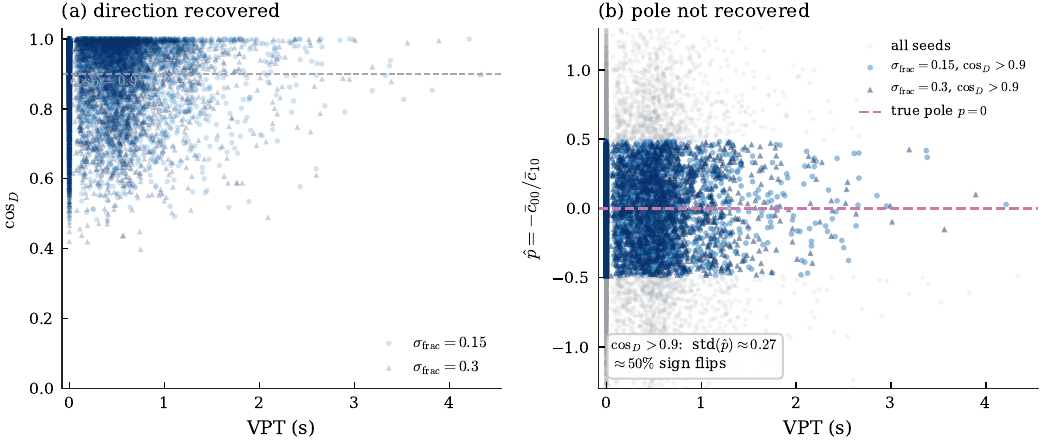}%
}
\caption{Denominator direction-versus-pole decoupling on noisy Lorenz ($s=x$).
  (a) The recovered denominator direction $\cos_D$ clusters high (median
  $\approx0.93$) at every valid prediction time. (b) The recovered pole of the
  affine denominator, $\hat p=-\bar c_{00}/\bar c_{10}$, scatters about its
  true value $0$ with standard deviation $\approx0.27$ and about $50\%$ sign
  reversals even when the direction is well recovered ($\cos_D>0.9$,
  highlighted). The averaging step pins the denominator direction but not the
  pole location.}
\label{fig:pole_decoupling}
\end{figure}

\subsection{Observability analysis}
\label{sec:obs}

To assess observability, we compute the
Letellier--Aguirre--Maquet coefficient~\cite{letellier2005relation}
\begin{equation}
\delta = \frac{\lambda_{\min}(J_\Phi^\top J_\Phi)}{\lambda_{\max}(J_\Phi^\top J_\Phi)}
  \in [0,1],
\label{eq:obs_coeff}
\end{equation}
where $J_\Phi=J_{\pi_J(\mathbf{U})}$ is the Jacobian of the differential-embedding map $\pi_J$
evaluated along the trajectory, mapping perturbations of the hidden state to
perturbations of the differential-embedding coordinates. The ratio compares the
weakest and strongest local stretching directions of this map: values near $1$
indicate that all state directions are captured comparably, whereas $\delta\ll1$
signals a nearly hidden direction and hence weak local observability.

We compute $J_\Phi$ along a trajectory of $500$ model time units ($50{,}000$ samples) for both
systems. For Lorenz with $s=x$, the median value is $\delta\approx
7.2\times10^{-6}$ and all sampled points lie below $10^{-3}$; within the region
$|x|<0.5$, which accounts for about $4.8\%$ of the trajectory, the median
decreases further to $\approx 1.2\times10^{-8}$. This region lies along the
analytic singular set $D^\star=3u_0=0$, so the same part of the attractor is both
closest to the closure pole and least observable. For R\"ossler with $s=x$, by
contrast, observability is far stronger: the median is $\delta\approx
1.3\times10^{-2}$ and only about $5.6\%$ of samples fall below $10^{-3}$, compared
with all of them for Lorenz. Its worst-observed region localizes near
$\{x=5.9\}$, where $D^\star\propto 500\,u_0-2950$ vanishes.

Figures~\ref{fig:lorenz_obs} and~\ref{fig:rossler_obs} compare $\log_{10}\delta$
across the coordinate observables $s\in\{x,y,z\}$ for the two systems. The full
screening (\ref{app:observability_screening} for Lorenz,
\ref{app:observability_screening_rossler} for R\"ossler) extends this to
twelve and eleven observables respectively: no tested Lorenz observable improves
the median $\delta$ by more than a factor of five over $s=x$, so within the
polynomial class the Lorenz reconstruction remains strongly nonuniform, whereas
R\"ossler admits a uniformly observable choice, $s=y$ ($\delta\equiv0.141$,
$|\det J_\Phi|\equiv1$).

\begin{figure}[H]
\centering
\scalebox{\figscale}{%
\includegraphics[width=\textwidth]{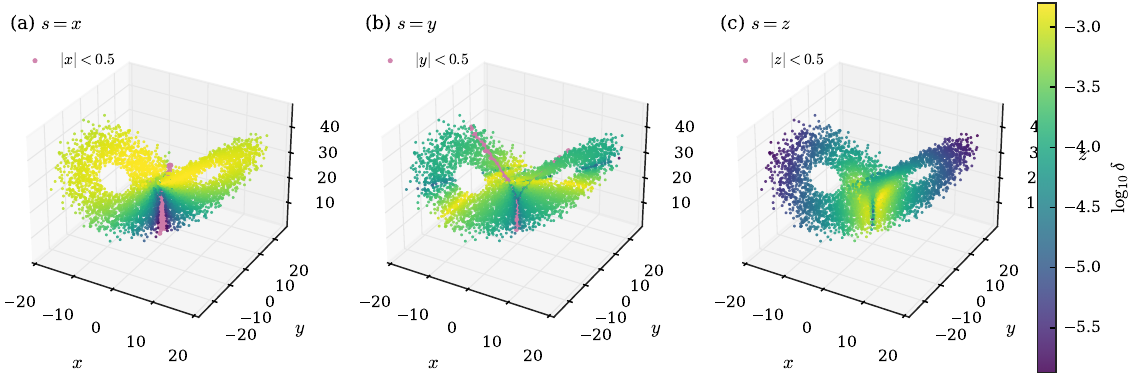}%
}
\caption{Lorenz observability coefficient $\log_{10}\delta$ for the scalar
  observables $s\in\{x,y,z\}$. Red regions indicate neighborhoods near $s=0$
  where observability is weakest. For the benchmark observable $s=x$, the
  least observable region lies along the analytic singular set
  $D^\star=3u_0=0$.}
\label{fig:lorenz_obs}
\end{figure}

\subsection{Comparison with prior forecasting methods}
\label{sec:comparison}

Table~\ref{tab:method_comparison} places our clean-data result alongside prior
methods for forecasting the continuation of the noise-free Lorenz $x$ time series,
extending the compilation of Wang and Guet~\cite{wang2021reconstructing}. Two axes
organize the comparison. The first is the observable: our method, together with
Wang and Guet and the scalar echo-state networks, forecasts from a single
coordinate, whereas reservoir computers and the LSTM use the full state
$(x,y,z)$. Because every entry in the table is a Lorenz forecast, we quote this
comparison in Lorenz time, as reported in the source studies. Among the
scalar-observable methods, our median prediction horizon ($11.6$ Lorenz time)
matches the best previously reported ($11.5$), and our best-realization horizon saturates the $20$-unit scoring window; because we measure the
horizon with a strict valid-prediction-time threshold ($0.2$ standard deviations,
sustained), it is if anything more conservative than the visual criterion used in
several of the cited studies. Typical full-state reservoir forecasts
($\sim$5--7 Lorenz time) fall below this, although specialized full-state methods
under noise-free, extended-precision conditions have recently reported much longer
horizons~\cite{gauthier2021next}; such results require the full state and idealized
data and are not comparable to the scalar-observable setting here.

The second axis is the model form. Reservoir and deep-network forecasters are
black boxes that return no governing equations, while sparse-regression methods
such as SINDy and its rational variants recover explicit equations but from
full-state data and do not report a Lorenz forecast
horizon~\cite{brunton2016discovering,mangan2016inferring}; deep delay
autoencoders recover sparse equations from a scalar series, but through a
trained neural coordinate transformation, so the recovered model lives in
learned latent coordinates rather than in fixed, analytically defined
embedding coordinates~\cite{bakarji2023discovering}; HAVOK yields a
linear-plus-forcing model in delay coordinates but is used to flag
regime-switching events rather than to forecast the
trajectory~\cite{brunton2017havok}. Our method occupies a distinct position: it
recovers an explicit rational closure from a scalar observable, validated both by
forecast skill and by direct coefficient-level comparison with the analytic
closure. It therefore matches the strongest scalar-observable forecasts while
additionally yielding interpretable governing equations in fixed embedding
coordinates. This comparison is for
the noise-free case, matching Table~\ref{tab:method_comparison}; the noisy-data
behavior of our method is reported in Section~\ref{sec:results}.

\begin{table}[t]
\centering\small
\caption{Prediction horizon for continuation of the noise-free Lorenz $x$ time
  series (units of Lorenz time). Scalar-observable methods use only $x$; full-state
  methods use $(x,y,z)$. Horizons are as
  reported under each study's own accuracy criterion. Because every entry is a Lorenz forecast, horizons are quoted in Lorenz time as reported in the source studies rather than in Lyapunov time; multiply by $\lambda_1=0.906$ to convert.}
\label{tab:method_comparison}
\resizebox{\linewidth}{!}{%
\begin{tabular}{llcl}
\toprule
Method & Details of the model & Obs. & Horizon \\
\midrule
\multicolumn{4}{l}{\textit{Scalar observable}}\\
Autoregressive~\cite{sauer1994time} & local linear model, delay coords & $x$ & $4.5$ \\
Neural difference eqn~\cite{ouala2020learning} & neural map, latent space & $x$ & $11$ \\
Neural ODE~\cite{ouala2020learning,chen2018neural} & adjoint training, latent space & $x$ & $0.5$ \\
Echo-state network$^{\dagger}$~\cite{mahata2023variability} & scalar-input ESN (median over ICs) & $x$ & $\sim$4.5$^{\ddagger}$ \\
Self-consistent (delay)~\cite{wang2021reconstructing} & Wang \& Guet, delay coords & $x$ & $8.5$ \\
Self-consistent (latent)~\cite{wang2021reconstructing} & Wang \& Guet, autoencoder latent & $x$ & $11.5$ \\
\textbf{This work} & rational closure, diff.-embedding coords; \emph{explicit eqns.} & $x$ & $11.6$\,(med.), $20$\,(best)$^{\ast}$ \\
\midrule
\multicolumn{4}{l}{\textit{Full state}}\\
LSTM~\cite{dubois2020data} & full observation + derivatives & $x,y,z$ & $3$ \\
Next-gen reservoir$^{\dagger}$~\cite{gauthier2021next} & polynomial (NVAR) reservoir & $x,y,z$ & $\sim$5.5$^{\ddagger}$ \\
Reservoir computer$^{\dagger}$~\cite{pathak2017replicate} & echo-state network & $x,y,z$ & $\sim$7$^{\ddagger}$ \\
\bottomrule
\end{tabular}%
}

\vspace{2pt}
{\footnotesize
$^{\ast}$ Valid prediction time at the $0.2$-std sustained threshold; the best value
saturates the $20$-unit scoring window. Unlike every other entry, this method also
returns an explicit rational closure.\\
$^{\ddagger}$ The cited study reports this as a qualitative divergence time, or in
Lyapunov times; values given in Lyapunov times are converted here using the Lorenz
Lyapunov time of $1.1$ model time units~\cite{gauthier2021next}.
}
\end{table}

\subsection{Limitations}

The main limitation is the fragility of the singular-set estimate to noise. The
recovery pins the denominator direction but not the pole location, so even modest
noise leaves the pole essentially unconstrained;
this is most acute for Lorenz with $s=x$, whose singular set passes through the
center of the attractor, the least observable region. Once the pole is mislocated
the forecast collapses well before the clean-data horizon, so the usable noise
level is low and strongly system-dependent rather than uniform. In this regime the recovery 
is fragile: small perturbations of the noisy denominator estimate shift the inferred pole
substantially, so the usable noise level is low and strongly system-dependent
rather than uniform. Coefficient accuracy degrades under noise as well, so for
Lorenz the forecast results are more informative than term-by-term coefficient
comparison.

\subsection{Future work}

Several extensions follow naturally. The reported benchmark fixes an affine
denominator ($d_D=1$) and a degree-four numerator ($d_{\max}=4$); these happen to
coincide with the analytic Lorenz and R\"ossler closures, so a principled
degree-selection criterion, for $d_D$ as much as for $d_N$, is needed before the
method can be applied to systems whose closure structure is unknown. The noise model
is restricted to additive Gaussian observation noise, and non-Gaussian, structured,
and multiplicative noise are natural to test next. Evaluation here is
within-trajectory, so cross-trajectory transfer tests would probe generalization.

Two routes could address the noise-fragility limitation more directly. 
A learned denoising front-end, such as a denoising
autoencoder~\cite{vincent2010stacked}, could regularize the noisy scalar record
before differentiation and supply a signal clean enough for the present pipeline
to recover a well-constrained singular set, in place of the spline smoothing used
here. This robustness would come at a computational cost, however: the Reinsch spline
used here is effectively instantaneous, whereas a neural denoiser must be trained
and is far more expensive to evaluate, so any gain in noise tolerance has to be
weighed against a substantial increase in runtime. Independently, constraining the
recovered closure so that $N/D$ remains bounded near the estimated singular set
$\{u_0=\hat p\}$ may improve forecast stability.

\section{Conclusion}
\label{sec:conclusion}

We have shown that a scalar-observable rational closure in differential embedding
coordinates can be identified directly from data by a weak-form recovery pipeline,
yielding an explicit and interpretable model. In the reported benchmark the
pipeline reaches the clean forecast horizons in the best cases on both Lorenz and
R\"ossler; and because the recovered coefficients are directly comparable with the
analytic closures, we can separate forecast skill from coefficient accuracy, since
under noise a forecast can remain useful even when individual coefficients are
poorly recovered. At $15$--$30\%$ additive noise the performance becomes strongly
system-dependent, with Lorenz supporting only short best-case forecasts and
R\"ossler longer ones in absolute time, though the ordering reverses in Lyapunov
times. The main limitation is the fragility of the singular-set
estimate to noise: because the recovery constrains the denominator direction but not the pole
location, even modest noise leaves the singular set essentially unconstrained,
most acutely for Lorenz, whose singular set lies in the least observable region of
the attractor. Overcoming this, for instance through a learned denoising front-end or a principled
choice of closure library, together with tests under broader noise models and
across trajectories, is a natural next step.

\section*{Acknowledgements}
A.S. thanks Professor Xavier Garbet (NTU) for his full support.

\section*{Funding}
This work was supported by the National Research Foundation, Singapore
[core funding ``Fusion Science for Clean Energy''].

\section*{CRediT authorship contribution statement}
\textbf{Ameir Shaa:} Conceptualization, Methodology, Software, Formal analysis,
Investigation, Visualization, Writing -- original draft, Writing -- review \& editing.
\textbf{Claude Guet:} Conceptualization, Supervision, Validation,
Writing -- review \& editing.

\section*{Declaration of competing interest}
The authors declare that they have no known competing financial interests or
personal relationships that could have appeared to influence the work reported
in this paper.

\section*{Data availability}
The code and data that reproduce the results and figures of this study are
available from the corresponding author upon reasonable request.


\newpage

\appendix
\makeatletter
\@addtoreset{figure}{section}
\@addtoreset{table}{section}
\makeatother
\setcounter{figure}{0}
\setcounter{table}{0}
\section{Benchmark settings}
\label{app:settings}

The two benchmark systems are the Lorenz~\cite{lorenz1963deterministic} system
\begin{equation}
\dot{x} = \sigma(y-x), \qquad
\dot{y} = x(\rho-z)-y, \qquad
\dot{z} = xy-\beta z,
\label{eq:lorenz}
\end{equation}
at the classical parameters $(\sigma,\rho,\beta)=(10,\,28,\,8/3)$, and the
R\"ossler~\cite{rossler1976equation} system
\begin{equation}
\dot{x} = -y-z, \qquad
\dot{y} = x+ay, \qquad
\dot{z} = b+z(x-c),
\label{eq:rossler}
\end{equation}
at $(a,b,c)=(0.2,\,0.2,\,5.7)$. Both systems are dimensionless: their model time
carries no physical unit, and all times below are quoted in model time units.

Table~\ref{tab:settings} collects the benchmark settings used throughout, shared
across systems and noise levels unless noted.

\begin{table}[H]
\centering\small
\caption{Benchmark settings. Systems are integrated at their classical
  parameters; all reconstruction, weak-form, integration, and scoring settings
  are shared across systems and noise levels unless noted. All times are in
  dimensionless model time units.}
\label{tab:settings}
\scalebox{0.71}{%
\begin{tabular}{ll}
\toprule
Setting & Value \\
\midrule
\multicolumn{2}{l}{\textit{Systems and observable}}\\
\quad Lorenz $(\sigma,\rho,\beta)$         & $(10,\,28,\,8/3)$ \\
\quad R\"ossler $(a,b,c)$                   & $(0.2,\,0.2,\,5.7)$ \\
\quad Scalar observable                     & $s=x$ \\
\midrule
\multicolumn{2}{l}{\textit{Data and sampling}}\\
\quad Sampling step $\Delta t$              & $0.01$ \\
\quad Initial state $\mathbf{U}(0)$         & $(1,0,0)$ (both systems) \\
\quad Transient                   & $50$ \\
\quad Noise $\sigma_{\mathrm{frac}}$        & $\{0,\,0.15,\,0.30\}$, \ $\sigma_\eta=\sigma_{\mathrm{frac}}\,\mathrm{std}(s)$ \\
\quad Lorenz $(t_{\mathrm{total}},H)$       & $(60,20)$ clean / $(60,10)$ noisy \\
\quad R\"ossler $(t_{\mathrm{total}},H)$    & $(200,100)$ clean / $(60,10)$ noisy \\
\quad Realizations per system per noise level                 & $5{,}000$ \\
\quad Initial-condition windows             & $8$ (interior) \\
\midrule
\multicolumn{2}{l}{\textit{Differential embedding}}\\
\quad Coordinates                           & $3$: $(u_0,u_1,u_2)$ \\
\quad FNN delay $\tau$ (Lorenz / R\"ossler) & $16$ / $84$ samples \\
\midrule
\multicolumn{2}{l}{\textit{Derivative estimation}}\\
\quad Spline                                & quintic Reinsch \\
\quad Penalty weight $\lambda$              & $N_{\mathrm{s}}\hat\sigma^2$ \\
\quad Noise scale $\hat\sigma$              & $1.4826\,\mathrm{MAD}(\Delta^2 y)/\sqrt{6}$ \\
\midrule
\multicolumn{2}{l}{\textit{Weak-form identification}}\\
\quad Test function $\psi(\xi)$             & $(1-\xi^2)^4$, \ $\xi\in[-1,1]$ \\
\quad Windows $K$ / half-width $R$          & $200$ / $20$ samples \\
\quad Denominator library $\Phi_D$          & affine $\{1,u_0\}$ \ ($d_D=1$) \\
\quad Numerator library $\Phi_N$            & deg $\le 4$ in $(u_0,u_1,u_2)$, $35$ monomials \\
\quad Denominator recovery                  & global SVD (clean); $12$-window average, $120$ samples/window (noisy) \\
\quad Numerator solve                       & ordinary least squares \\
\midrule
\multicolumn{2}{l}{\textit{Forecast integration}}\\
\quad Integrator                            & DOP853 \\
\quad $(\mathrm{rtol},\mathrm{atol})$       & $(10^{-8},10^{-10})$ \\
\quad \texttt{max\_step}                    & $0.05$ \\
\midrule
\multicolumn{2}{l}{\textit{VPT scoring}}\\
\quad Error threshold $\varepsilon_{\mathrm{thr}}$ & $0.2$ \\
\quad Sustained run                         & $20$ consecutive samples \\
\quad Normalization                         & $\mathrm{std}(u_0^{\mathrm{true}}\vert_{[0,H]})$ \\
\bottomrule
\end{tabular}%
}
\end{table}

\section{Local Reconstruction in Three Dimensions}
\label{app:jet_proof}

\begin{theorem}[3D differential-embedding reconstruction]
\label{thm:jet_invertibility}
Let $\mathbf{U}(t)\in\mathbb{R}^3$ satisfy
\begin{equation}
\dot{\mathbf{U}} = \mathbf{F}(\mathbf{U}),
\label{eq:app_thm_system}
\end{equation}
where $\mathbf{F}$ is smooth, and let $\pi:\mathbb{R}^3\to\mathbb{R}$ be a
smooth scalar observable. Define the observed signal
$s(t)=\pi(\mathbf{U}(t))$ and the differential-embedding map
\begin{equation}
\pi_J(\mathbf{U}) = \bigl(s,\dot{s},\ddot{s}\bigr).
\label{eq:app_jet_map_general}
\end{equation}
If
\begin{equation}
\det J_{\pi_J(\mathbf{U})} \neq 0,
\label{eq:app_jet_nondeg}
\end{equation}
then $\pi_J$ is locally invertible at $\mathbf{U}$. Equivalently, the
differential embedding coordinates
\begin{equation}
\mathbf{u}=(u_0,u_1,u_2)=(s,\dot{s},\ddot{s})
\label{eq:app_coords}
\end{equation}
locally determine the state, and there exists a smooth local closure function
$\Phi$ such that
\begin{equation}
\dot{u}_0 = u_1, \qquad
\dot{u}_1 = u_2, \qquad
\dot{u}_2 = \Phi(u_0,u_1,u_2).
\label{eq:app_jet_closure_general}
\end{equation}
\end{theorem}

\begin{proof}
In the three-dimensional square setting, the map $\pi_J:\mathbb{R}^3\to
\mathbb{R}^3$ has the Jacobian
\begin{equation}
J_{\pi_J(\mathbf{U})} =
\begin{pmatrix}
\nabla s(\mathbf{U}) \\
\nabla \dot{s}(\mathbf{U}) \\
\nabla \ddot{s}(\mathbf{U})
\end{pmatrix},
\label{eq:app_jet_jacobian}
\end{equation}
whose rows are the gradients of the three differential embedding coordinates.
Condition~\eqref{eq:app_jet_nondeg} implies that $J_{\pi_J(\mathbf{U})}$ is
invertible. The inverse function theorem therefore gives a local inverse
$\pi_J^{-1}$ in a neighbourhood of $\mathbf{u}=\pi_J(\mathbf{U})$. Hence the
state can be reconstructed locally from $(u_0,u_1,u_2)$.

The time derivatives of the first two coordinates satisfy
\begin{equation}
\dot{u}_0 = u_1, \qquad \dot{u}_1 = u_2
\label{eq:app_proof_lower}
\end{equation}
by definition. For the third coordinate, the vector field gives a smooth
evolution law
\begin{equation}
\dot{u}_2 = G(\mathbf{U})
\label{eq:app_proof_G}
\end{equation}
for some smooth function $G$. Composing with the local inverse
$\mathbf{U}=\pi_J^{-1}(\mathbf{u})$ yields
\begin{equation}
\dot{u}_2 = G\bigl(\pi_J^{-1}(\mathbf{u})\bigr) = \Phi(\mathbf{u}),
\label{eq:app_proof_phi}
\end{equation}
which is the claimed local closure.
\end{proof}

\section{Rational Closure Construction \& Polynomial Approximation via Stone--Weierstrass}
\label{app:sw}

This appendix explains why the closure $\Phi$ in differential embedding
coordinates~\eqref{eq:jet_ode} admits a rational representation and why a
polynomial rational ansatz is justified on the compact embedded attractor.
We work in the general setting of Sec.~\ref{sec:theory}. The rational representation below holds whenever $\pi_J$ is locally invertible. Note that in the following derivation, we assume that the Jacobian is a square matrix.

\subsection{Rational closure by construction}\label{app:rationalclosure}

Differentiating $\mathbf{u}=\pi_J(\mathbf{U})$ along trajectories and using
$\dot{\mathbf{U}}=\mathbf{F}(\mathbf{U})$ from~\eqref{eq:system} gives the linear
relation
\begin{equation}
\dot{\mathbf{u}} = J_{\pi_J(\mathbf{U})}\,\mathbf{F}(\mathbf{U}),
\label{eq:app_sw_chain}
\end{equation}
where $J_{\pi_J(\mathbf{U})}$ denotes the Jacobian matrix of the
differential-embedding map $\pi_J$ at $\mathbf{U}$.

Applying Cramer's rule~\cite{horn2012matrix} for inverting the matrix $J_{\pi_J(\mathbf{U})}$,

\begin{equation}
\operatorname{adj}\!\bigl(J_{\pi_J(\mathbf{U})}\bigr)\,\dot{\mathbf{u}}
=
\det\!\bigl(J_{\pi_J(\mathbf{U})}\bigr)\,\mathbf{F}(\mathbf{U}).
\label{eq:app_sw_adj}
\end{equation}

This is a linear relation among the components of
$\dot{\mathbf{u}}=(\dot{u}_0,\ldots,\dot{u}_k)$. The \emph{degeneracy set}
$\mathcal{D}=\{\mathbf{U}:\det J_{\pi_J(\mathbf{U})}=0\}$
marks where $\pi_J$ loses local invertibility.
Away from the degeneracy set $\mathcal{D}$, Equation~\eqref{eq:app_sw_adj} may be regarded as a linear system for the unknown vector $\dot{\mathbf u}$. Applying Cramer’s rule a second time to this system yields an explicit expression for the highest jet rate.

To make this explicit, let
\[
A(\mathbf{U})
=
\operatorname{adj}\!\bigl(J_{\pi_J(\mathbf{U})}\bigr),
\qquad
\mathbf{b}(\mathbf{U})
=
\det\!\bigl(J_{\pi_J(\mathbf{U})}\bigr)\mathbf{F}(\mathbf{U}),
\]
so that~\eqref{eq:app_sw_adj} may be written as
\[
A(\mathbf{U})\dot{\mathbf{u}}=\mathbf{b}(\mathbf{U}).
\]
Let \(A_k(\mathbf{U})\) denote the matrix obtained from \(A(\mathbf{U})\)
by replacing its \((k+1)\)-st column by \(\mathbf{b}(\mathbf{U})\).
Cramer's rule then gives
\begin{equation}
\dot{u}_k
=
\frac{\det A_k(\mathbf{U})}
{\det A(\mathbf{U})}.
\label{eq:app_sw_cramer}
\end{equation}
Since
\[
\det\!\bigl(\operatorname{adj}(M)\bigr)
=
\det(M)^{n-1}
\]
for an \(n\times n\) matrix \(M\), the denominator is
\[
\det A(\mathbf{U})
=
\det\!\bigl(J_{\pi_J(\mathbf{U})}\bigr)^{n-1}.
\]
Moreover, the factor
\(\det(J_{\pi_J(\mathbf{U})})\) may be extracted from the replaced
column of \(A_k(\mathbf{U})\). Hence
\begin{equation}
\dot{u}_k
=
\frac{
P(\mathbf{U})
}{
\det\!\bigl(J_{\pi_J(\mathbf{U})}\bigr)^{n-2}
},
\label{eq:app_sw_ratio_general}
\end{equation}
where
\begin{equation}
P(\mathbf{U})
=
\det\!\left[
A_0(\mathbf{U}),\ldots,
A_{k-1}(\mathbf{U}),
\mathbf{F}(\mathbf{U})
\right],
\label{eq:app_sw_numerator}
\end{equation}
and \(A_j(\mathbf{U})\) denotes the \(j\)-th column of
\(A(\mathbf{U})\).

In the three-dimensional case \(n=3\), so that
\(\mathbf{u}=(u_0,u_1,u_2)\), equation~\eqref{eq:app_sw_ratio_general}
reduces to
\begin{equation}
\dot{u}_2=u_3
=
\frac{P(\mathbf{U})}
{\det\!\bigl(J_{\pi_J(\mathbf{U})}\bigr)}.
\label{eq:app_sw_local_ratio_U}
\end{equation}
Away from \(\mathcal{D}\), the differential-embedding map is locally
invertible. Substituting
\(\mathbf{U}=\pi_J^{-1}(\mathbf{u})\) therefore gives
\begin{equation}
u_3
=
\frac{
P\!\left(\pi_J^{-1}(\mathbf{u})\right)
}{
\det\!\left(
J_{\pi_J(\pi_J^{-1}(\mathbf{u}))}
\right)
}.
\label{eq:app_sw_local_ratio}
\end{equation}

\subsection{Polynomial rational ansatz}

Since $\mathcal{A}$ is compact and $\pi_J$ is continuous, the embedded
attractor $\mathcal{A}_J=\pi_J(\mathcal{A})$ is compact~\cite{rudin1976principles}.

Away from the degeneracy set $\mathcal{D}$, \ref{app:rationalclosure} shows that the
highest jet rate may be written locally as
\begin{equation}
\dot{u}_k
=
\frac{N(\mathbf{u})}{D(\mathbf{u})},
\label{eq:closure_local}
\end{equation}
where
\[
D(\mathbf{u})
=
\det\!\left(
J_{\pi_J}\!\left(\pi_J^{-1}(\mathbf{u})\right)
\right)
\]
and $N(\mathbf{u})$ is obtained from the corresponding Cramer numerator.
Multiplying by $D(\mathbf{u})$ removes the denominator and yields the rational
closure ansatz
\begin{equation}
D(\mathbf{u})\,\dot{u}_k
=
N(\mathbf{u}).
\label{eq:closure}
\end{equation}

Although the quotient~\eqref{eq:closure_local} becomes singular on
$\mathcal{D}$, the product~\eqref{eq:closure} is regular, since both
$D(\mathbf{u})$ and $N(\mathbf{u})$ are locally smooth functions of the
jet coordinates away from $\mathcal{D}$. Two regimes justify representing these functions by
polynomials.

\paragraph{Approximation regime (general smooth systems).}

For a general smooth vector field and smooth observable, both
$D(\mathbf{u})$ and $N(\mathbf{u})$ are continuous on the compact set
$\mathcal{A}_J$. By the Stone--Weierstrass theorem~\cite{rudin1976principles},
for every $\varepsilon>0$ there exist polynomials
$p_D(\mathbf{u})$ and $p_N(\mathbf{u})$ such that
\begin{equation}
\sup_{\mathbf{u}\in\mathcal{A}_J}
|D(\mathbf{u})-p_D(\mathbf{u})|
<
\varepsilon,
\qquad
\sup_{\mathbf{u}\in\mathcal{A}_J}
|N(\mathbf{u})-p_N(\mathbf{u})|
<
\varepsilon.
\label{eq:app_sw_stw}
\end{equation}
Consequently,
\begin{equation}
D(\mathbf{u})\approx p_D(\mathbf{u}),
\qquad
N(\mathbf{u})\approx p_N(\mathbf{u}),
\label{eq:app_sw_poly}
\end{equation}
yielding a polynomial approximation to the rational closure
\eqref{eq:closure}. The approximation becomes exact in the limit of increasing
polynomial degree.

\paragraph{Exact regime (polynomial dynamics and observable).}

If $\mathbf{F}$ and the observable $\pi$ are polynomial, then the jet
components, the Jacobian
$J_{\pi_J(\mathbf{U})}$, its determinant, and the associated Cramer numerator
are all polynomial in $\mathbf{U}$. Consequently, after composition with the
local inverse $\pi_J^{-1}$, both $D(\mathbf{u})$ and $N(\mathbf{u})$ are
polynomial, and the rational closure~\eqref{eq:closure} is exact at finite
degree.

\medskip
The point of this appendix is not to claim that the exact closure $\Phi$ is
globally polynomial.
Rather, it shows that $\Phi$ is rational wherever $\pi_J$ is locally
invertible, that the pole structure is encoded in
$\det J_{\pi_J(\mathbf{U})}$, and that Stone--Weierstrass justifies
polynomial approximation of both $D$ and the regularized numerator on the
compact attractor $\mathcal{A}_J$.

\section{Noise-scale estimate from second differences}
\label{app:noise_scale}

We derive the noise-scale estimate
$\hat\sigma = 1.4826\,\mathrm{MAD}(\Delta^2 y)/\sqrt{6}$ of
Eq.~\eqref{eq:noise_scale} from the i.i.d.\ $\mathcal{N}(0,\sigma_\eta^2)$ noise model
of Section~\ref{sec:method}, so that $\hat\sigma$ recovers $\sigma_\eta$ from the data
alone, with no prior knowledge of the noise level.

The second-difference operator $\Delta^2 y_i = y_{i+1} - 2y_i + y_{i-1}$ annihilates
any locally affine part of $y_i$. Writing $y_i = s(t_i)+\eta_i$ and Taylor-expanding
the smooth signal gives
$\Delta^2 s(t_i) = \ddot{s}(t_i)\,\Delta t^2 + \mathcal{O}(\Delta t^4)$, which is
negligible at small $\Delta t$, so the differenced sequence is dominated by the noise
combination $\Delta^2\eta_i = \eta_{i+1} - 2\eta_i + \eta_{i-1}$. As a weighted sum of
independent Gaussians with weights $(1,-2,1)$,
\begin{equation}
\begin{aligned}
  \Delta^2\eta_i &\sim \mathcal{N}\!\bigl(0,\,6\sigma_\eta^2\bigr),\\
  \mathrm{Var}(\Delta^2\eta_i) &= \bigl(1^2+(-2)^2+1^2\bigr)\sigma_\eta^2 = 6\,\sigma_\eta^2,
\end{aligned}
\end{equation}
so its standard deviation is $\sqrt{6}\,\sigma_\eta$ and dividing by $\sqrt{6}$ returns
the scale of $\sigma_\eta$ itself.

It remains to estimate that standard deviation robustly, which we do with the median
absolute deviation $\mathrm{MAD}(z)=\mathrm{median}_i\,|z_i-\mathrm{median}_j\,z_j|$
\cite{donoho1994ideal}, insensitive to the heavy tails left by outliers or by any
signal surviving the differencing. For a zero-mean Gaussian $X\sim\mathcal{N}(0,\varsigma^2)$
the inner median vanishes, so $\mathrm{MAD}(X)=\mathrm{median}\,|X|$ is the value $m$
with $\Pr(|X|\le m)=\tfrac12$. Since $\Pr(|X|\le m)=2\varphi(m/\varsigma)-1$ with $\varphi$
the standard-normal CDF,
\begin{equation}
\begin{aligned}
  2\varphi\!\left(\frac{m}{\varsigma}\right)-1&=\frac12
  \;\Longrightarrow\;
  m=\varsigma\,\varphi^{-1}\!\left(\tfrac34\right),\\
  \varsigma&=\frac{\mathrm{MAD}(X)}{\varphi^{-1}(3/4)}\approx 1.4826\,\mathrm{MAD}(X),
\end{aligned}
\end{equation}
using $\varphi^{-1}(3/4)\approx 0.6745$. Applying this to $\Delta^2 y$ estimates its
standard deviation $\sqrt{6}\,\sigma_\eta$ as $1.4826\,\mathrm{MAD}(\Delta^2 y)$, and
the final division by $\sqrt{6}$ gives $\hat\sigma\approx\sigma_\eta$, which is
Eq.~\eqref{eq:noise_scale}. On clean data the dispersion of $\Delta^2 y$ collapses,
$\hat\sigma\to 0$, and the Reinsch penalty $\lambda=N_{\mathrm{s}}\hat\sigma^2$
vanishes so the spline interpolates the record.

\section{Null-space solution of the homogeneous closure system}
\label{app:nullspace}

The closure coefficients satisfy the homogeneous system $M\,\theta = 0$
[Eq.~\eqref{eq:strong_identity} stacked over samples], with
$\theta=[\theta_D^\top,\theta_N^\top]^\top$. A common rescaling of $\theta$ leaves
the closure $\Phi=N/D$ unchanged, so $\theta$ is fixed only up to scale; we impose
$\lVert\theta\rVert_2=1$ and seek the unit vector that $M$ maps closest to zero,
\begin{equation}
\theta^\star = \arg\min_{\lVert\theta\rVert_2=1}\lVert M\theta\rVert_2^2
            = \arg\min_{\lVert\theta\rVert_2=1}\theta^\top M^\top M\,\theta.
\label{eq:nullspace_obj}
\end{equation}
Let $M=U\Sigma V^\top$ be the singular value decomposition, with $U,V$ orthogonal
and $\Sigma=\mathrm{diag}(\sigma_1\ge\cdots\ge\sigma_n\ge0)$. Since $U$ is
orthogonal it preserves the norm, so with $c=V^\top\theta$ (also unit-norm),
\begin{equation}
\lVert M\theta\rVert_2^2 = \lVert \Sigma\,c\rVert_2^2 = \sum_{i=1}^{n}\sigma_i^2\,c_i^2 .
\end{equation}
Minimizing $\sum_i\sigma_i^2 c_i^2$ subject to $\sum_i c_i^2=1$ places all weight on
the smallest singular value, $c=e_n$, giving the minimum value $\sigma_{\min}^2$ at
\begin{equation}
\theta^\star = V e_n = v_n ,
\end{equation}
the right singular vector associated with $\sigma_{\min}$. If $M$ is exactly
rank-deficient ($\sigma_{\min}=0$) then $Mv_n=0$ and $v_n$ spans the null space;
under noise $\sigma_{\min}>0$ and $v_n$ is the best approximate null vector. By
singular-vector perturbation theory~\cite{horn2012matrix}, a perturbation $\delta M$
rotates $v_n$ by an amount controlled by the inverse gap
$1/(\sigma_{\mathrm{next}}-\sigma_{\min})$ to the next singular value; a small gap
makes the recovered direction ill-conditioned.

\section{Observability Screening for Scalar Observables}
\label{app:observability_screening}

We screen scalar observables for both benchmark systems by the Letellier
observability coefficient
$\delta=|\lambda_{\min}(J_\Phi^\top J_\Phi)|/|\lambda_{\max}(J_\Phi^\top J_\Phi)|$,
ranking by the median $\delta$ along a trajectory of $500$ model time units ($50{,}000$ samples,
DOP853, $\mathrm{rtol}=10^{-10}$, $\mathrm{atol}=10^{-12}$). The Lorenz screening
is given first, then the R\"ossler screening.


\subsection{Lorenz Scalar Observables}

\begin{table}[H]
  \centering
  \small
  \caption{Letellier observability coefficient $\delta$ across $12$ Lorenz scalar
  observables, ranked by median $\delta$ (descending). No observable improves on
  the $s=x$ baseline ($\delta_{\mathrm{median}}\approx7.2\times10^{-6}$) by the
  screening criterion of a $\geq10\times$ median gain with uniform attractor
  coverage (p$_{95}$/p$_1<10$); the best candidate, $s=z$, gives only about
  $5\times$ and remains highly nonuniform.}
  \label{tab:observability_alternatives}
  \resizebox{\linewidth}{!}{%
  \begin{tabular}{r|l|r|r|r|r|r}
    \hline
    \textbf{Rank} & \textbf{Observable} & \textbf{Median $\delta$} & \textbf{p$_1$ $\delta$} & \textbf{p$_{95}$ $\delta$} & \textbf{Frac $\delta <10^{-3}$} & \textbf{Uniform} \\
    \hline
    1  & $s=z$                    & $3.63 \times 10^{-5}$ & $8.03 \times 10^{-7}$ & $5.29 \times 10^{-4}$ & 0.973 & False \\
    2  & $s=x+y+z$                & $2.50 \times 10^{-5}$ & $1.33 \times 10^{-8}$ & $1.61 \times 10^{-4}$ & 1.000 & False \\
    3  & $s=y+z$                  & $1.87 \times 10^{-5}$ & $3.43 \times 10^{-8}$ & $2.32 \times 10^{-4}$ & 0.996 & False \\
    4  & $s=x+z$                  & $8.56 \times 10^{-6}$ & $2.39 \times 10^{-9}$ & $4.85 \times 10^{-5}$ & 1.000 & False \\
    5  & $s=x^2+y^2+(z-27)^2$     & $8.23 \times 10^{-6}$ & $7.99 \times 10^{-10}$ & $8.43 \times 10^{-4}$ & 0.962 & False \\
    6  & $s=x$ (baseline)          & $7.17 \times 10^{-6}$ & $1.62 \times 10^{-9}$ & $1.93 \times 10^{-5}$ & 1.000 & False \\
    7  & $s=x+y$                  & $5.92 \times 10^{-6}$ & $2.83 \times 10^{-9}$ & $9.99 \times 10^{-5}$ & 1.000 & False \\
    8  & $s=y$                    & $4.89 \times 10^{-6}$ & $2.05 \times 10^{-9}$ & $1.51 \times 10^{-4}$ & 1.000 & False \\
    9  & $s=x^2-z$                & $2.02 \times 10^{-6}$ & $2.98 \times 10^{-9}$ & $4.29 \times 10^{-5}$ & 1.000 & False \\
    10 & $s=xy$                   & $1.17 \times 10^{-6}$ & $1.13 \times 10^{-10}$ & $3.84 \times 10^{-5}$ & 1.000 & False \\
    11 & $s=x^2+y^2$              & $4.56 \times 10^{-7}$ & $9.11 \times 10^{-11}$ & $6.98 \times 10^{-6}$ & 1.000 & False \\
    12 & $s=xyz$                  & $3.37 \times 10^{-7}$ & $2.59 \times 10^{-11}$ & $1.51 \times 10^{-5}$ & 1.000 & False \\
    \hline
  \end{tabular}%
  }
\end{table}

No candidate in this set meets both screening criteria: a tenfold median-$\delta$
improvement over the $s=x$ baseline and uniform attractor coverage. The best
median belongs to $s=z$ (about $5\times$ the baseline), but this masks severe
nonuniformity: the median $\delta$ over the region $|x|<0.5$ is roughly
$657\times$ smaller than the attractor-wide median for $s=x$, so the worst-region
pathology persists even for the best-ranked observable. The lower-ranked
observables (rows~9--12) reach smaller medians only through extreme localization,
with near-zero observability over most of the attractor, and are unsuitable for
general closure recovery. Within the polynomial observables tested, the Lorenz
observability limitation is therefore structural: changing the observable alone
does not produce uniform observability, so the noisy-regime ceiling reflects
observability geometry and initial-condition sensitivity rather than the
derivative estimator alone.

\subsection{R\"ossler Scalar Observables}
\label{app:observability_screening_rossler}

\begin{table}[H]
  \centering
  \small
  \caption{Letellier observability coefficient $\delta$ across $11$ R\"ossler
  scalar observables, ranked by median $\delta$ (descending). In contrast to
  Lorenz (Table~\ref{tab:observability_alternatives}), R\"ossler admits a
  \emph{uniformly} observable polynomial observable, $s=y$ ($\delta\equiv0.141$,
  $|\det J_\Phi|\equiv1$); the benchmark observable $s=x$ is moderately observable
  and $s=z$ is poorly observable.}
  \label{tab:observability_rossler}
  \resizebox{\linewidth}{!}{%
  \begin{tabular}{r|l|r|r|r|r|r}
    \hline
    \textbf{Rank} & \textbf{Observable} & \textbf{Median $\delta$} & \textbf{p$_1$ $\delta$} & \textbf{p$_{95}$ $\delta$} & \textbf{Frac $\delta <10^{-3}$} & \textbf{Uniform} \\
    \hline
    1  & $s=y$ (uniform)          & $1.41 \times 10^{-1}$ & $1.41 \times 10^{-1}$ & $1.41 \times 10^{-1}$ & 0.000 & True \\
    2  & $s=x+y$                  & $2.82 \times 10^{-2}$ & $2.16 \times 10^{-4}$ & $1.08 \times 10^{-1}$ & 0.029 & False \\
    3  & $s=x$ (baseline)         & $1.28 \times 10^{-2}$ & $8.63 \times 10^{-6}$ & $8.04 \times 10^{-2}$ & 0.056 & False \\
    4  & $s=xy$                   & $1.08 \times 10^{-2}$ & $6.65 \times 10^{-6}$ & $1.69 \times 10^{-2}$ & 0.106 & False \\
    5  & $s=x^2-z$                & $1.13 \times 10^{-3}$ & $8.19 \times 10^{-7}$ & $7.23 \times 10^{-3}$ & 0.476 & False \\
    6  & $s=x^2+y^2$              & $1.00 \times 10^{-3}$ & $3.17 \times 10^{-7}$ & $6.29 \times 10^{-2}$ & 0.499 & False \\
    7  & $s=x+y+z$                & $3.41 \times 10^{-4}$ & $1.07 \times 10^{-5}$ & $4.53 \times 10^{-2}$ & 0.623 & False \\
    8  & $s=y+z$                  & $2.57 \times 10^{-4}$ & $1.67 \times 10^{-5}$ & $3.88 \times 10^{-2}$ & 0.659 & False \\
    9  & $s=x+z$                  & $1.83 \times 10^{-4}$ & $2.58 \times 10^{-6}$ & $2.96 \times 10^{-2}$ & 0.689 & False \\
    10 & $s=xyz$                  & $2.23 \times 10^{-7}$ & $3.40 \times 10^{-10}$ & $5.55 \times 10^{-4}$ & 0.973 & False \\
    11 & $s=z$                    & $6.74 \times 10^{-9}$ & $1.61 \times 10^{-11}$ & $3.86 \times 10^{-4}$ & 0.975 & False \\
    \hline
  \end{tabular}%
  }
\end{table}

\begin{figure}[H]
\centering
\scalebox{\figscale}{%
\includegraphics[width=\textwidth]{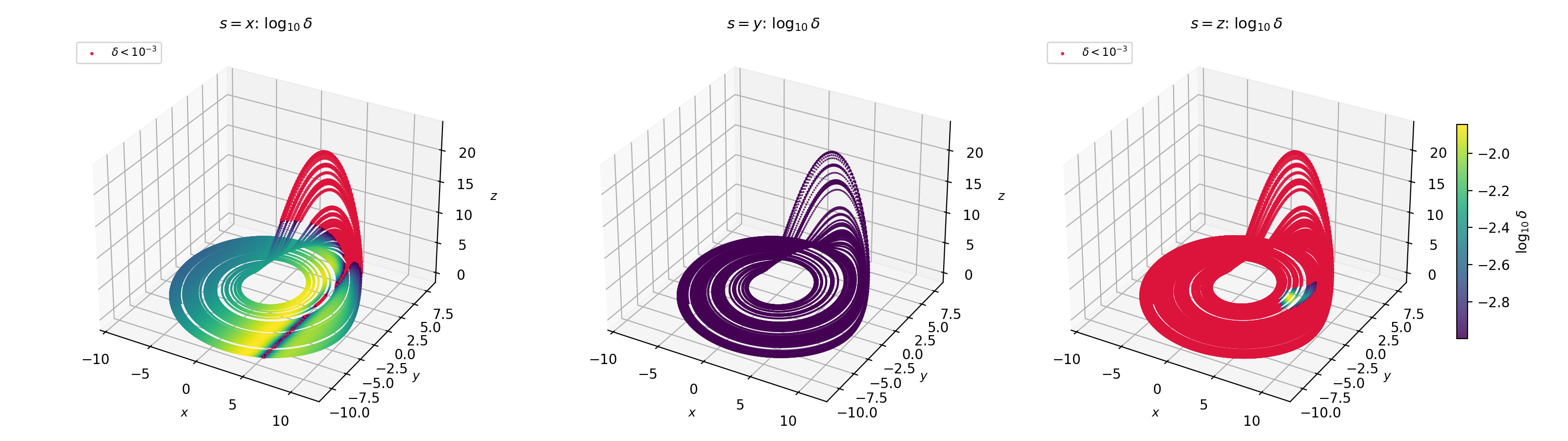}%
}
\caption{R\"ossler observability coefficient $\log_{10}\delta$ for the scalar
  observables $s\in\{x,y,z\}$; points with $\delta<10^{-3}$ are highlighted in
  crimson. The coefficient is constant for $s=y$, so no weak region appears; for
  the benchmark $s=x$ the weak region localizes near $x=5.9$; for $s=z$
  observability is weak across almost the entire attractor.}
\label{fig:rossler_obs}
\end{figure}

In observability terms R\"ossler is essentially the inverse of Lorenz: no Lorenz
polynomial observable is uniform, whereas R\"ossler admits one, $s=y$, with a
constant observability Jacobian across the attractor. The benchmark $s=x$ is only
moderately observable, with a weak region near $x=5.9$, and $s=z$ is weak almost
everywhere. This stronger, more uniform observability is consistent with the
longer noisy-regime forecasts reported for R\"ossler in the main text.

\section{DOP853 Integrator Tolerance and the
  \texorpdfstring{\\*}{}Lyapunov-Amplification Ceiling}
\label{app:tol}

This appendix is only intended as a heuristic scale check for the clean Lorenz median; it is
not an independent proof of the benchmark result. For an ODE integrated with
DOP853 at relative tolerance $\mathrm{rtol}$, the accumulated integration error
after $N_{\mathrm{step}}$ steps grows as
$\varepsilon_0 \sim N_{\mathrm{step}}\cdot\mathrm{rtol}$ in the worst
case~\cite{hairer1993solving}. 
On the Lorenz attractor with $\Delta t=10^{-2}$ and the clean-data regime horizon
$H=20$ in model time, $N_{\mathrm{step}} \approx 2\,000$, giving $\varepsilon_0 \approx 2\times10^{-5}$
at the deployed $\mathrm{rtol}=10^{-8}$.
With the absolute VPT threshold
$\varepsilon_{\mathrm{thr}}\,\mathrm{std}(x) \approx 1.75$ (with
$\varepsilon_{\mathrm{thr}}=0.2$), an initial error $\varepsilon_0$ amplified as
$e^{\lambda_1 t}$ reaches the threshold after a time
$T_{\mathrm{VPT}} \approx \lambda_1^{-1}\ln[\varepsilon_{\mathrm{thr}}\,\mathrm{std}(x)/\varepsilon_0]$.
Measured in Lyapunov times the growth rate drops out, and the
Lyapunov-amplification ceiling is simply
\begin{equation}
  \Lambda_{\mathrm{VPT}} = \lambda_1 T_{\mathrm{VPT}} \approx
    \ln\!\frac{\varepsilon_{\mathrm{thr}}\,\mathrm{std}(x)}{\varepsilon_0}.
\label{eq:vpt_ceiling}
\end{equation}
This gives a rough worst-case estimate of order $\Lambda \approx 11.4$.
The clean Lorenz median VPT ($10.5$ Lyapunov times, Sec.~\ref{sec:results}) is of the same
order, and the best realizations saturate the $18.1$-Lyapunov-time horizon; realizations whose trajectory
avoids the $D=0$ pole during the scored horizon exceed this crude estimate
because the adaptive integrator usually accumulates less error than the
worst-case bound.

R\"{o}ssler has a Lyapunov ceiling well
beyond its $7.1$-Lyapunov-time horizon, so on clean data the global-SVD pin succeeds and
the forecast is integrator-limited rather than amplification-limited, saturating
the horizon; the heavy tail in Table~\ref{tab:vpt} appears only in the noisy
regime, reflecting the difficulty of pinning the off-center pole under noise,
not this ceiling.

\bibliographystyle{elsarticle-num}
\bibliography{references}

\end{document}